\documentclass[a4paper,10pt]{article}

\usepackage{lcbmaths}
\usepackage{graphicx}
\usepackage{lmodern}
\usepackage[T1]{fontenc}
\usepackage{microtype}      % microtypography
\usepackage{enumitem}
\usepackage{wrapfig}
\usepackage[colorlinks,citecolor=blue,pdftex,hypertexnames=false]{hyperref}
\usepackage[top=1.5cm,bottom=2.5cm,left=2cm,right=1.9cm]{geometry}
\usepackage{natbib}
\newcommand\exop      {\mathbb{E}}
\newcommand\enop      {\mathbf{H}}
\newcommand\trop      {\mathrm{\scalebox{0.65}{\textsf{T}}}}
\newcommand\rrop      {\mathrm{\scalebox{0.65}{\textsf{R}}}}
\newcommand\rropb     {\mathrm{\scalebox{0.75}{\textsf{R}}}}
\newcommand\eps       {\varepsilon}
\renewcommand\phi     {\varphi}
\newcommand\dt        {\Delta}
\newcommand\beps      {\boldsymbol\eps}
\newcommand\bbeta     {\boldsymbol\eta}

\newcommand\obeps     {\bar\beps}
\newcommand\bx        {\boldsymbol{x}}
\newcommand\by        {\boldsymbol{y}}
\newcommand\bz        {\boldsymbol{z}}

\newcommand\bv        {\boldsymbol{v}}
\newcommand\bw        {\boldsymbol{w}}
\newcommand\be        {\boldsymbol{e}}
\newcommand\bY        {\boldsymbol{Y}}
\newcommand\bA        {\boldsymbol{A}}
\newcommand\bB        {\boldsymbol{B}}
\newcommand\bE        {\boldsymbol{E}}

\newcommand\oby       {\bar\by}
\newcommand\obz       {\bar\bz}

\newcommand\oA        {\bar{A}}
\newcommand\oB        {\bar{B}}

\newcommand\oK        {\bar{K}}

\newcommand\oP        {\bar{P}}

\newcommand\oSig      {\bar\Sigma}

\newcommand\cR        {\mathcal{R}}
\newcommand\cE        {\mathcal{E}}
\newcommand\cF        {\mathcal{F}}
\newcommand\cT        {\mathcal{T}}
\newcommand\cG        {\mathcal{G}}

\newcommand\cN        {\mathcal{N}}

\newcommand\vbar      {{\,|\,}}

\newcommand\Kay       {\scalebox{1.5}{$\kappa$}} % big kappa!
\newcommand\oKay      {\bar\Kay}

\newcommand\ndist[2]  {\mathcal N\!\left({#1},{#2}\right)}
\newtheorem{theorem}{Theorem}
\newtheorem{corollary}{Corollary} % Numbered by section (e.g., 1.1)
\newcommand{\footremember}[2]{%
    \footnote{#2}
    \newcounter{#1}
    \setcounter{#1}{\value{footnote}}%
}
\newcommand{\footrecall}[1]{%
    \footnotemark[\value{#1}]%
}

\begin{document}

\title{Granger Causality Maps for Langevin Systems}

\author{Lionel Barnett\footremember{SCCS}{Sussex Centre for Consciousness Science, Department of Informatics, University of Sussex, Falmer, Brighton, UK}\footremember{corrauth}{Corresponding author: \texttt{\href{mailto:l.c.barnett@sussex.ac.uk}{l.c.barnett@sussex.ac.uk}}}, \ Benjamin Wahl\footremember{iresearcher}{Unaffiliated\vspace{4pt}}, \ Nadine Spychala\footrecall{SCCS} \ and Anil K. Seth\footrecall{SCCS}}

\date\today

\maketitle

\begin{abstract}
\citet{WahlEtal:2016,WahlEtal:2017} introduced the idea of Granger causality (GC) maps for Langevin systems: dynamics are localised linearly at each point in phase space as vector Ornstein-Uhlenbeck (VOU) processes, for which GCs may in principle be calculated, thus constructing a GC map on phase space. Their implementation, however, suffered a significant drawback: GCs were approximated from models based on discrete-time stroboscopic sampling of local VOU processes, which is not only computationally inefficient, but more seriously, infeasible on regions of phase space where local dynamics are unstable, leaving ``holes'' in the GC maps. We solve these problems by deriving an analytical expression for GC rates associated with a VOU process which, under quite general conditions, yields a meaningful solution even in the unstable case. Applied to GC maps, this not only ``fills in the holes'', but also furnishes a computationally efficient method of calculation devolving to solution of continuous-time algebraic Riccati equations which, in the case of a univariate source, become simple quadratic equations.  We show, furthermore, that the GC rate for VOU processes is invariant under rescaling of the overall fluctuations intensity, so that GC maps may effectively be calculated for \emph{deterministic} nonlinear dynamical systems, with a residual ``ghost of noise'' represented by a variance-covariance map.
\end{abstract}


\section{Introduction\label{sec:intro}}

Wiener-Granger causality (henceforth GC), a widely-used method for quantifying directed information transfer between stochastic variables, is based on the premise that cause (a) precedes effect, and (b) contains unique information about effect \citep{Wiener:1956,Granger:1963}. While physics, and indeed other branches of science, are traditionally concerned with ``mechanism'', in the sense of the structure and parameters of models, the appeal of information theory is that it abstracts away mechanism in favour of causal (in the Wiener-Granger sense) relationships among system variables. As such, it has been applied in fields as diverse as econometrics, the neurosciences, genomics, ecology and climate science. Most commonly operationalised via linear modelling \citep{Geweke:1982,Geweke:1984}, it is widely (if sometimes unfairly\footnote{It is underappreciated that even if a (stationary) stochastic process has a nonlinear generative mechanism, it may nonetheless be amenable to linear modelling. Specifically, under mild conditions it will have a Wold representation as a linear moving-average which may, if minimum-phase conditions hold, be suitable for Granger-causal analysis.}) viewed as inappropriate for stochastic systems featuring nonlinear interactions -- in contrast to its nonparametric cousin transfer entropy \citep[TE;][]{Schreiber:2000,Palus:2001,BarnettSeth:2009}. Transfer entropy is, however, in general analytically intractable for nonlinear systems, especially in continuous time \citep{SpinneyEtal:2017}, and frequently problematic to estimate empirically \citep{ShahsavariBaboukaniEtal:2020}.

Of especial interest are \emph{diffusion processes} described by Langevin equations---equivalently (in general nonlinear) stochastic differential equations (SDEs) or Fokker-Planck equations---which are ubiquitous in statistical physics and have further applications in biology, econometrics, machine learning and beyond. \citet{WahlEtal:2016,WahlEtal:2017} present a powerful approach to analysis of information transfer in autonomous Langevin systems. Following standard practice when faced with analytically intractable nonlinearity, they linearise locally in the system phase space. The locally linearised dynamics of a Langevin process may, under mild assumptions, be represented by a vector Ornstein-Uhlenbeck (VOU) process. Noting that stroboscopic (regularly spaced) observation of a VOU process yields a 1st-order vector autoregressive (VAR) process, now in discrete time, \citet{WahlEtal:2016} calculate Granger causalities between a given pair of sub-processes from local VAR(1) models obtained from suitably fine-grained subsampling, thus constructing a mapping from phase space to local GC values for those sub-processes. This mapping may, further, be averaged over the stationary distribution of the process in phase space, yielding a global, system-wide GC value for the sub-processes in question.

This procedure, however, has two drawbacks. Firstly, we note that as the sampling interval shrinks, the GC value approaches zero \citep{FlorensFougere:1996}, while the GC \emph{rate}---information transfer per unit time---approaches a finite limit \citep{BarnettSeth:ctgc:2017}. Since the sample time interval in the subsampling procedure is finite, the procedure can thus yield at best approximations to local GCs rates. Secondly, and more problematically, it is commonplace in nonlinear dynamics that, even if globally stable, dynamics may well be locally unstable in large regions of the phase space (indeed, local instability is a fundamental characteristic of chaotic dynamics). Locally linearised dynamics are determined by the Jacobian: the system is locally stable at a point in phase space iff the Jacobian matrix evaluated at that point is Hurwitz-stable; that is, all its eigenvalues lie strictly in the left half-plane in the complex plane. The problem, then, is that available techniques for deriving GCs from local discrete-time VAR(1) models, \eg, via spectral factorisation \citep{Wilson:1972,Dhamala:2008a}, fail in the unstable case, leaving ``holes'' in phase space where the GC map is undefined. A particular consequence of these holes, is that it becomes impossible to calculate meaningful system-wide GC values. \citet{WahlEtal:2016} propose a scheme whereby local GCs are averaged only over stable regions in phase space; however the resulting average fails to reflect the contribution to global information transfer of unstable regions.

In this article we resolve both issues, firstly by demonstrating that GC rates for a VOU process\footnote{I.e., global GC rates; as VOU processes are linear, there is nothing to gain from local linearisation -- see Section~\ref{sec:lvougc}.} may be calculated directly from the VOU model parameters without recourse to subsampling, and that, furthermore, under broad conditions results retain an interpretion as local information transfer even when the local VOU model is unstable. Our resolution leans heavily on previous work by the authors: an analytic formulation of GC rates for a class of continuous-time, distributed-lag  stochastic processes \citep{BarnettSeth:ctgc:2017}, and calculation of Granger causalities for discrete-time state-space systems \citep{BarnettSeth:ssgc:2015}.

We also demonstrate a seldom-remarked invariance of GC with respect to rescaling of the noise intensity. This has interesting consequences for GC maps, namely that it permits an extension of information transfer analysis to \emph{deterministic} nonlinear dynamical systems described by sets of ordinary differential equations (ODEs), by globally dialling the noise down to zero, while retaining its ``ghost'' in the form of a noise variance-covariance map on the phase space, which we interpret as specifying notional infinitesimal fluctuations.

In summary, the main contributions of this paper are: (i) calculation of GC rates for VOU processes directly from model parameters, including in the unstable case, allowing (ii) construction of GC maps for Langevin systems over the entire phase space, and  (iii) the extension of GC maps and global GC rates to classical deterministic dynamics.

\section{Granger causality rate for vector Ornstein-Uhlenbeck processes} \label{sec:vougc}

\citet{BarnettSeth:ctgc:2017} define the \emph{zero-horizon Granger causality rate} (henceforth GC rate) for a class of distributed-lag vector stochastic processes in continuous time. Although the analysis there only addresses the unconditional case, all results extend straightforwardly to the conditional case. We start with a brief recap of the construction.

We consider the class of stationary, zero-mean, continuous-time vector moving-average (VMA) processes of the form
\begin{equation}
	\by(t) = \int_{u = 0}^\infty B(u) \,d\bw(t-u) \,,\qquad -\infty < t < \infty\,. \label{eq:ctma}
\end{equation}
Here $\bw(t)$ is an $n$-dimensional Wiener process with $d\bw(t) \sim \ndist 0{\Sigma\,dt}$, where $\Sigma$ is an $n \times n$ positive-definite covariance matrix, and $B(u)$ an $n \times n$ square-integrable moving-average kernel, with $B(0) = I$. The integral in \eqref{eq:ctma} is to be interpreted as an It\=o integral \citep{Oksendal:2003}. It is assumed that the VMA form \eqref{eq:ctma} may be inverted to yield a vector autoregressive (VAR) form as a stochastic integro-differential equation
\begin{equation}
	d\by(t) = \left[\int_{u = 0}^\infty A(u)\by(t-u)\,du\right]\!dt + d\bw(t) \,, -\infty < t < \infty \label{eq:ctar}
\end{equation}
with square-integrable autoregressive kernel $A(u)$. \citet{BarnettSeth:ctgc:2017} consider only \emph{stable} and \emph{minimum-phase} models, and it is further assumed that any \emph{sub-}process of $\by(t)$ also has an invertible VMA representation. We refer to such processes as continuous-time vector autoregressive (CTVAR) processes.

Given a CTVAR process $\by(t)$ as above, \citet{BarnettSeth:ctgc:2017} consider the optimal least-squares predictor $\hat\by(t;h) = \exop[\by(t+h) \vbar \by(u) : u \le t]$ at finite prediction horizon $h > 0$, and show that the covariance matrix of the (stationary) prediction error process $\be(t;h) = \hat\by(t;h)-\by(t+h)$ may be expressed as
\begin{equation}
	\cE(h) = \exop\big[\be(t;h)\be(t;h)^\trop\big] = \int_{u=0}^h B(u) \Sigma B(u)^\trop \,du\,, \label{eq:cE}
\end{equation}
where superscript $\trop$ denotes matrix transpose. Suppose now that $\by$ is partitioned as\footnote{Throughout, subscripts $1,2,3$ on vectors and matrices are multi-indices corresponding to the given partitioning.} $\by = [\by_1^\trop\; \by_2^\trop\; \by_3^\trop]^\trop$. For GC analysis, $\by_1$ will be the \emph{target} variable, $\by_2$ the \emph{conditioning} variable and $\by_3$ the \emph{source} variable. Throughout, we use subscript $\rropb$ to denote the multi-index pair $(1,2)$, and superscript $\rropb$ to denote quantities associated with the ``reduced system'' $\by_\rrop = [\by_1^\trop\; \by_2^\trop]^\trop$; \ie, with the source variable $\by_3$ omitted. Analogous to the discrete-time case \citep{Geweke:1982,Geweke:1984}, the Granger causality from $\by_3$ to $\by_1$ conditional on $\by_2$ at prediction horizon $h$ is defined as
\begin{equation}
	\cF_{\by_3 \to \by_1 | \by_2}(h) = \log\frac{|\cE^\rrop_{11}(h)|}{|\cE_{11}(h)|}, \label{eq:GCh}
\end{equation}
where $|\cdots|$ denotes matrix determinant. It is shown that $\cF_{\by_3 \to \by_1 | \by_2}(0) = 0$, and the GC rate is defined as
\begin{equation}
	\cR_{\by_3 \to \by_1 | \by_2} = \dot\cF_{\by_3 \to \by_1 | \by_2}(0) = \lim_{h \to 0}\, \frac1h \cF_{\by_3 \to \by_1 | \by_2}(h), \label{eq:GCrdef}
\end{equation}
where the dot denotes time differentiation. Finally, it is shown that
\begin{equation}
	\cR_{\by_3 \to \by_1 | \by_2} = \traceop\!\Big[\Sigma^{-1}_{11}\big(D^\rrop_{11} - D_{11}\big)\Big], \label{eq:gcrate}
\end{equation}
where
\begin{equation}
	D = \tfrac12 \ddot\cE(0) = \tfrac12\big[\dot B(0)\Sigma + \Sigma \dot B(0)^\trop  \big] \label{eq:ddef}
\end{equation}
and $D^\rrop$ is the corresponding quantity for the reduced process.

Vector Ornstein-Uhlenbeck processes, as defined by a linear SDE of the form:
\begin{equation}
	d\by(t) = A \by(t) \,dt + d\bw(t) \label{eq:vou}
\end{equation}
are a special case of CTVAR processes with autoregressive kernel $A(u) = A\delta(u)$, where $A$ is an $n \times n$ matrix. Stability requires that $|Iz-A| \ne 0$ for $z$ in the right half-plane $\re(z) \ge 0$ of the complex plane---\ie, $A$ is Hurwitz-stable---and the process is always minimum-phase. A sub-process of a stable VOU process, while not in general itself a VOU process, will have (stable, minimum-phase) VMA and VAR representations of the form \eqref{eq:ctma} and \eqref{eq:ctar} respectively. \citet[Appendix~F]{BarnettSeth:ctgc:2017} show that in the general case $\dot B(u) = \int_0^u A(s) B(u-s) \,ds$, $u > 0$, from which we derive $B(u) = e^{Au}$ for the VOU process \eqref{eq:vou}, leading to $\dot B(0) = A$ and
\begin{equation}
	D = \tfrac12 \big(A\Sigma + \Sigma A^\trop\big). \label{eq:D1}
\end{equation}

Relaxing the stability requirement (\ie, $A$ may have eigenvalues in the right complex half-plane), the process $\by(t)$ may no longer be assumed stationary, and consequently may not, as in \eqref{eq:ctma}, be taken to extend into the infinite past. Thus, rather than \eqref{eq:ctma}, we consider CTVAR processes of the form
\begin{equation}
	\by(t) = \int_{u = 0}^t B(u) \,d\bw(t-u) \,, \qquad t \ge 0\,, \label{eq:ctma0}
\end{equation}
initialised at $t = 0$, where $B(u)$ is no longer assumed square-integrable. In Appendix~\ref{apx:fpestat} we show that \eqref{eq:ctma0} may always be inverted to yield a (not necessarily stable) continuous-time VAR representation, and although the process $\by(t)$ itself may not be stationary, the finite-horizon prediction error process $\be(t;h)$ is nonetheless (wide-sense) stationary with covariance matrix $\cE(h)$ as in \eqref{eq:cE}. The construction of $\cR_{\by_3 \to \by_1 | \by_2}$ outlined above thus goes through unchanged, and---noting that $\by_\rrop(t)$ will not in general be a VOU process---it remains to calculate $\dot B^\rrop(0)$, and thence $D^\rrop$. Our principal result shows how $\cR_{\by_3 \to \by_1 | \by_2}$ may be calculated explicitly from the  VOU parameters $(A,\Sigma)$ under relaxed assumptions on the stability of $A$ (we shall still require $\Sigma$ to be positive-definite).\\

\begin{theorem} \label{the:main}
For the VOU process \eqref{eq:vou} with $d\bw(t) \sim \ndist 0{\Sigma\,dt}$ and a partitioning $\by = [\by_1^\trop\; \by_2^\trop\; \by_3^\trop]^\trop$, if $\Sigma$ is positive-definite and the matrix pair $(A_{33},\, A_{\rrop3})$ detectable, then the continuous-time algebraic Riccati equation (CARE)
\begin{equation}
	A_{33}P_{33} + P_{33} A_{33}^\trop + \Sigma_{33} = \big(P_{33} A_{\rrop3}^\trop + \Sigma_{3\rrop}\big) \Sigma_{\rrop\rrop}^{-1} \big(P_{33} A_{\rrop3}^\trop + \Sigma_{3\rrop}\big)^\trop \label{eq:scare}
\end{equation}
has a unique stabilising solution $P_{33}$, and the Granger causality rate from $\by_3$ to $\by_1$ conditional on $\by_2$ for the VOU \eqref{eq:vou} is given by:
\begin{equation}
	\cR_{\by_3 \to \by_1|\by_2} = \traceop\!\Big[\Sigma^{-1}_{11} A_{13} P_{33} A_{13}^\trop\Big], \label{eq:gcrate1}
\end{equation}
and is equal to twice the corresponding transfer entropy rate, under an appropriate definition of the latter (\cf~\citet{BarnettSeth:2009,Barnett:teml:2012}).
\end{theorem}
\begin{proof}
See Appendix~\ref{apx:vougcproof}.
\end{proof}

\noindent We also state an unconditional counterpart to Theorem~\ref{the:main}:\\

\begin{corollary}
Given the partitioning $\by = [\by_1^\trop\; \by_2^\trop\; \by_3^\trop]^\trop$, then assuming the associated CAREs have unique stabilising solutions, the unconditional Granger causality rate from $\by_2$ to $\by_1$ for the VOU \eqref{eq:vou} is given by:
\begin{equation}
	\cR_{\by_2 \to \by_1} = \cR_{\by_{23} \to \by_1} - \cR_{\by_3 \to \by_1|\by_2} \label{eq:gcrun}
\end{equation}
where $\by_{23} = [\by_2^\trop\; \by_3^\trop]^\trop$.
\end{corollary}
\begin{proof}
This follows from the corresponding standard result in discrete time \citep[Sec.~3]{Geweke:1984}, which survives passage to the continuous-time limit \cite[Sec.~3.3,~eq.~59]{BarnettSeth:ctgc:2017}.
\end{proof}
\noindent We note the following:
\begin{enumerate}
\item Under the weaker condition that $\Sigma$ is positive-semidefinite and $\Sigma_{\rrop\rrop}$ invertible, the CARE \eqref{eq:scare} has a unique stabilising solution under the additional condition that $(A_{33}-\Sigma_{3\rrop}\Sigma_{\rrop\rrop}^{-1}A_{\rrop3},\, \Sigma_{33}-\Sigma_{\rrop3}\Sigma_{\rrop\rrop}^{-1}\Sigma_{3\rrop})$ is stabilisable; see Appendix~\ref{apx:vougcproof}. It is not entirely clear however, how to interpret the resulting GC rate $\cR_{\by_3 \to \by_1|\by_2}$, so we do not consider this case further.
\item If either or both of $A$, $A_{33}$ are Hurwitz-stable, then the detectability condition of Theorem~\ref{the:main} holds. The case $A_{33}$ stable is trivial; the case $A$ stable is proved in Appendix~\ref{apx:PBH}.
\item Of the VAR coefficients, only $A_{13}, A_{23}, A_{33}$---those with $\by_3$ as source---appear in the solution (by contrast, all the $\Sigma_{ij}$ potentially affect the result). In particular, $\cR_{\by_3 \to \by_1|\by_2}$ vanishes as expected\footnote{In the particular case that $A_{\rrop3}$ (and thus $A_{13}$) are identically zero and $A_{33}$ is not Hurwitz-stable, while the detectability condition of Theorem~\ref{the:main} is violated it still makes sense to declare $\cR_{\by_3 \to \by_1|\by_2} = 0$, since the dynamics of $\by_\rrop$ are completely unaffected by the source $\by_3$.} if $A_{13} = 0$.
\item If $\dim(\by_3) = 1$, \ie, the source variable is $1$-dimensional, then $P_{33}$ is scalar and the CARE \eqref{eq:scare} becomes the quadratic equation in $P_{33}$:
\begin{equation}
	\big(A_{\rrop3}^\trop \Sigma_{\rrop\rrop}^{-1} A_{\rrop3}\big) P_{33}^2 - 2 \big( A_{33} - \Sigma_{3\rrop} \Sigma_{\rrop\rrop}^{-1} A_{\rrop3}\big) P_{33} - \big(\Sigma_{33} - \Sigma_{3\rrop}\Sigma_{\rrop\rrop}^{-1}\Sigma_{\rrop3}\big) = 0\,.  \label{eq:qcare}
\end{equation}
The unique stabilising solution for $P_{33}$ (if it exists) corresponds to the solution of \eqref{eq:qcare} with the positive square root of the discriminant\footnote{More precisely, there are two scenarios: (i) the vector $A_{\rrop3}$ is identically zero, in which case we require $A_{33} < 0$ (the Granger causality rate is then zero); or (ii) $A_{\rrop3}$ is not identically zero, in which case $A_{\rrop3}^\trop \Sigma_{\rrop\rrop}^{-1} A_{\rrop3} > 0$, and the discriminant of the quadratic equation \eqref{eq:qcare} is nonnegative, as is the positive square root solution.}.
\item\label{it:rescale} If we \emph{rescale} the residuals covariance matrix by $\Sigma \to \nu\Sigma$, then $\nu P_{33}$ is the solution of the corresponding CARE \eqref{eq:scare}, and from \eqref{eq:gcrate1} we see that the Granger causality rate $\cR_{\by_3 \to \by_1|\by_2}$ is unchanged. Thus the Granger causality rate is invariant under rescaling of the overall intensity (but not in general under changes to relative variances/covariances) of fluctuations.
\end{enumerate}

\noindent An important special case is the Granger-causal graph \citep{Seth:2008}; \ie, the pairwise-conditional causality rates
\begin{equation}
	\cG_{ij} = \cR_{y_j \to y_i | \by_{[ij]}}\,, \qquad i,j = 1,\ldots,n, \; i \ne j \label{eq:gcgraph1}
\end{equation}
where subscript $[\cdots]$ indicates that the enclosed indices are omitted. Since the source variable is $1$-dimensional, we may apply \eqref{eq:qcare}. The sub-process indices $1,2,3$ and $\rrop$ in the previous analysis then map as $1 \to i$, $2 \to [ij]$, $3 \to j$ and $\rrop \to [j]$, and we have
\begin{equation}
	\cG_{ij} = \Sigma^{-1}_{ii} A_{ij}^2 P_{jj} \label{eq:gcgraph}
\end{equation}
with $P_{jj}$ the (positive root) solution of the quadratic equation
\begin{equation}
	\Big(A_{[j] j}^\trop \Sigma_{[j][j]}^{-1} A_{[j] j}\Big) P_{jj}^2
	- 2 \Big(A_{jj} - \Sigma_{j[j]} \Sigma_{[j][j]}^{-1} A_{[j] j}\Big) P_{jj}
	- \Big(\Sigma_{jj} - \Sigma_{j[j]} \Sigma_{[j][j]}^{-1} \Sigma_{[j] j}\Big) = 0\,. \label{eq:carecg}
\end{equation}
For the unconditional case, from \eqref{eq:gcrun} we have $\cR_{y_j \to y_i} = \cR_{y_{[i]} \to y_i} - \cG_{ij}$. For $\cR_{y_{[i]} \to y_i}$ the sub-process indices map as $1 \to i$, $2 \to \emptyset$, $3 \to [i]$ and $\rrop \to i$, and the corresponding (n-1)-dimensional CARE is
\begin{equation}
	A_{[i][i]}P_{[i][i]} + P_{[i][i]}A_{[i][i]}^\trop + \Sigma_{[i][i]} = \big(P_{[i][i]} A_{i[i]}^\trop + \Sigma_{[i]i}\big) \Sigma_{ii}^{-1} \big(P_{[i][i]} A_{i[i]}^\trop + \Sigma_{[i]i}\big)^\trop \,. \label{eq:careucg}
\end{equation}
We then have
\begin{equation}
	\cR_{y_j \to y_i} = \Sigma^{-1}_{ii}\traceop\!\Big[A_{i[i]}P_{[i][i]}A_{i[i]}^\trop\Big] - \cG_{ij}\,.
\end{equation}

\section{Granger causality maps for Langevin processes\label{sec:lvougc}}

Following \citet{WahlEtal:2016}, we consider multivariate Langevin systems specified by autonomous SDEs of the form
\begin{equation}
	d\by(t) = f\big(\by(t)\big)dt + d\bw\big(\by(t),t\big) \label{eq:nlsde}
\end{equation}
on $\reals^n$, with \emph{drift function} $f : \reals^n \to \reals^n$, and Wiener noise $d\bw(\by,t) \sim \ndist 0{\Sigma(\by)dt}$, where $\Sigma(\by)$ is the \emph{diffusion function}, which maps from $\reals^n$ to the manifold of $n \times n$ positive-definite matrices. Again following \citet{WahlEtal:2016}, we linearise \eqref{eq:nlsde} around a point $\by_0 \in \reals^n$. Setting $\bbeta(t) = \by(t)-\by_0$ we assume $\Vert\bbeta(t)\Vert < \eps$, and work to $\bigOsym(\eps)$. We have
\begin{equation}
	d\by(t) = f\big(\by_0 + \bbeta(t)\big) \,dt + d\bw\big(\by_0 + \bbeta(t),t\big)
\end{equation}
Now
\begin{equation}
	f\big(\by_0 + \bbeta(t)\big) = f(\by_0) + J(\by_0) \cdot \bbeta(t) + \bigOsym\big(\eps^2\big)
\end{equation}
where $J(\by_0) = \nabla\!f(\by_0)$ is the \emph{Jacobian matrix}\footnote{We assume all requisite derivatives of $f(\by)$ and $\Sigma(\by)$ exist.} of $f(\cdots)$ evaluated at $\by_0$, and for $\Vert\bbeta\Vert < \eps$
\begin{equation}
	d\bw(\by_0+\bbeta,t) \sim \ndist 0{\left\{\Sigma(\by_0) + \nabla\Sigma(\by_0) \cdot \bbeta + \bigOsym\big(\eps^2\big)\right\}dt}
\end{equation}
We make the further assumption that the fluctuations covariance term $\Sigma(\by)$ changes slowly with $\by$\footnote{\citet{WahlEtal:2016} refer to this assumption, common in statistical physics, as ``weakly multiplicative noise''.}; specifically, we assume $\Vert\nabla\Sigma(\by)\Vert  = \bigOsym(\eps)$ everywhere. We then have to $\bigOsym(\eps)$
\begin{equation}
	d\bbeta(t) = \left\{f(\by_0) + J(\by_0) \cdot \bbeta(t)\right\}dt + d\bw(\by_0,t)
\end{equation}
Setting
\begin{equation}
	\bz(t) = J(\by_0)^{-1} \cdot f(\by_0)+\bbeta(t) \label{eq:llvar}
\end{equation}
(a linear translation in $\reals^n$ to the mean reversion level), we find that to first order in $\eps$, $\bz(t)$ satisfies the VOU
\begin{equation}
	d\bz(t) = J(\by_0) \cdot \bz(t) \,dt + d\bw(\by_0,t) \label{eq:llvou}
\end{equation}
The analysis of the previous Section may be applied to the locally linearised process $\bz(t)$ around $\by_0$, wherever the Jacobian matrix $J(\by_0)$ is nonsingular. Note that local stability requires invertibility of $J(\by_0)$, since its eigenvalues must lie in the complex half-plane $\re(z) < 0$. Here we don't demand stability of the locally linearised VOUs everywhere, but assume that the Jacobian is singular on at most a set of measure zero in $\reals^n$.

Given a partitioning of the phase space into $\by = [\by_1^\trop\; \by_2^\trop\; \by_3^\trop]^\trop$ as in Section~\ref{sec:vougc}, the mapping $\by_0 \mapsto \cR_{\bz_3 \to \bz_1 | \bz_2}(\by_0)$ as $\by_0$ varies over the system phase space, with $\bz$ as in \eqref{eq:llvar}, defines the Granger causality map. Note that the map inherits invariance under rescaling of the covariance matrices $\Sigma(\by) \to \nu(\by)\Sigma(\by)$. Conceptually, the map might thus be regarded as associated with the \emph{deterministic} autonomous ordinary differential equation (ODE) $\dot\by(t) = f\big(\by(t)\big)$ and a given equivalence class of covariance matrix mappings $\by \mapsto \Sigma(\by)$ under rescaling.

As suggested by \cite{WahlEtal:2016}, we may obtain a \emph{global} GC rate, as the expectation of $\cR_{\bz_3 \to \bz_1 | \bz_2}(\by_0)$ over the stationary distribution (assuming it exists) of states defined by the dynamics of \eqref{eq:nlsde}. Under ergodicity assumptions, in the limit of high noise this tends to the distribution of the Wiener process $\bw\big(\by,t\big)$ appearing in \eqref{eq:nlsde}. In the limit of low noise, we might, alternatively, take the expectation over the \emph{attractor(s)} of the autonomous ODE, leaving the covariance function $\Sigma(\by)$ as a ``ghost'' of notional noise fluctuations. Note that by noise scale invariance, the only influence of noise level on the global GC rate is the distribution on phase space with respect to which we calculate the expectation. Under ergodicity assumptions, we may in practice calculate the global GC rate as
\begin{equation}
	\cR_{\by_3 \to \by_1 | \by_2} = \lim_{T \to \infty} \frac1T \int_{t=0}^T \cR_{\bz_3 \to \bz_1 | \bz_2}\big(\by(t)\big) \,dt \label{eq:gcint}
\end{equation}
where $\by(t)$ is a trajectory of \eqref{eq:nlsde}. \citet{WahlEtal:2016} show that if the process \eqref{eq:nlsde} is actually linear (\ie, it is a VOU), then the global GC rate always corresponds to the VOU Granger causality rate $\cR_{\by_3 \to \by_1|\by_2}$ of Theorem~\ref{the:main}. In the nonlinear case, we stress that, although the (nonparametric) transfer entropy rate from $\by_3$ to $\by_1$ conditioned on $\by_2$ may be well-defined \citep{SpinneyEtal:2017}, we have no reason to  expect that the global GC rate \eqref{eq:gcint} will correspond to the TE rate for any covariance mapping $\Sigma(\by)$; in other words, while the equivalence of GC with TE applies locally in phase space (as per Theorem~\ref{the:main}), it may not be assumed to hold globally for global GC defined, as described above, as the averaged local GC.\\

\subsection{Example: the Lorenz system} \label{sec:lorenz}

The well-known three-variable Lorenz system \citep{Lorenz:1963} is defined by the parametrised set of ODEs:
\begin{subequations}
\begin{align}
	\dot y_1 &= \sigma(y_2-y_1) \\
	\dot y_2 &= y_1(\rho-y_3)-y_2 \\
	\dot y_3 &= y_1y_2-\beta y_3
\end{align} \label{eq:lorenz}%
\end{subequations}
In some parameter regimes it exhibits chaotic dynamics, with the iconic ``butterfly'' strange attractor.

We posit a set of SDEs of the form \eqref{eq:nlsde} based on \eqref{eq:lorenz}, with constant fluctuations covariance matrix $\Sigma(\by) = \nu I$. Fluctuations are then globally scaled down by letting $\nu \to 0$, so that in the limit the (deterministic) dynamics are those of the ODEs \eqref{eq:lorenz}. The Jacobian of the system is
\begin{equation}
	J(\by) = \begin{bmatrix}
		-\sigma & \sigma & 0 \\
		\rho - y_3 & -1 & -y_1 \\
		y_2 & y_1 & -\beta
	\end{bmatrix}
\end{equation}
We have
\begin{equation}
	|J(\by)| = \sigma [\beta(\rho-1-y_3) - y_1(y_1+y_2)]
\end{equation}
which vanishes on the quadratic ($2$-dimensional and hence measure-zero) surface $y_3 = \rho-1-\frac1\beta y_1(y_1+y_2)$.

\begin{wrapfigure}{r}{0pt}
  \includegraphics{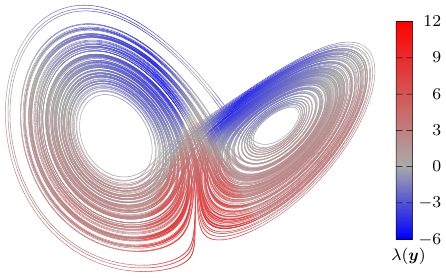}
  \caption{\small Local stability of Lorenz dynamics. The colour scale corresponds to the largest real part  $\lambda(\by)$ of the eigenvalues of the Jacobian matrix $J(\by)$; the system is locally unstable where $\lambda(\by) \ge 0$ (grey$\,\to\,$red).} \label{fig:lorenz1}
\end{wrapfigure} \ \\
We simulated the deterministic equations \eqref{eq:lorenz} with canonical parameters $\sigma = 10, \rho = 28, \beta = 8/3$, using the Runge-Kutta $(4,5)$ method implemented by the MATLAB\textsuperscript{\textregistered} function \texttt{ode45}, for $200$ seconds\footnote{The time unit is arbitrary labelled as ``seconds''.}, sampling the trajectory at time increments of $0.01$ seconds\footnote{The \texttt{ode45} solver uses an adaptive step size which may not correspond to the sampling increment $h$.}, with initial value $\by(0) = [1\;1\;1]^\trop$, and allowing $100$ seconds for the trajectory to settle into the attractor. Figure.~\ref{fig:lorenz1} displays the largest real part  $\lambda(\by)$ of the eigenvalues of $J(\by)$ plotted on  the attractor; the system is locally unstable where $\lambda(\by) \ge 0$ (grey$\,\to\,$red).
\iffalse
\begin{figure}[t!]
	\begin{center}
		\includegraphics{lorenz_sim_stability.pdf}
	\end{center}
	\caption{\small Local stability of Lorenz dynamics on the attractor. The colour scale corresponds to the largest real part  $\lambda(\by)$ of the eigenvalues of the Jacobian matrix $J(\by)$; the system is locally unstable where $\lambda(\by) \ge 0$ (grey$\,\to\,$red).} \label{fig:lorenz1}
\end{figure}
\fi
We then calculated the Granger causal graph $\cG_{ij}(\by)$ \eqref{eq:gcgraph1} on the attractor (Figure~\ref{fig:lorenz2}), estimating global GC rates according to \eqref{eq:gcint} by numerical quadrature. We may confirm from the form of $J(\by)$ that, since the $J_{ii}(\by)$ are always $< 0$, the the detectability condition of Theorem~\ref{the:main} is always satisfied. Since $J_{13}(\by) = 0$, $\cG_{13}(\by) \equiv 0$ as expected [\cf~\eqref{eq:gcgraph}], while all other pairwise-conditional GC rates are generally nonzero.
\begin{figure*}[t!]
	\begin{center}
		\includegraphics{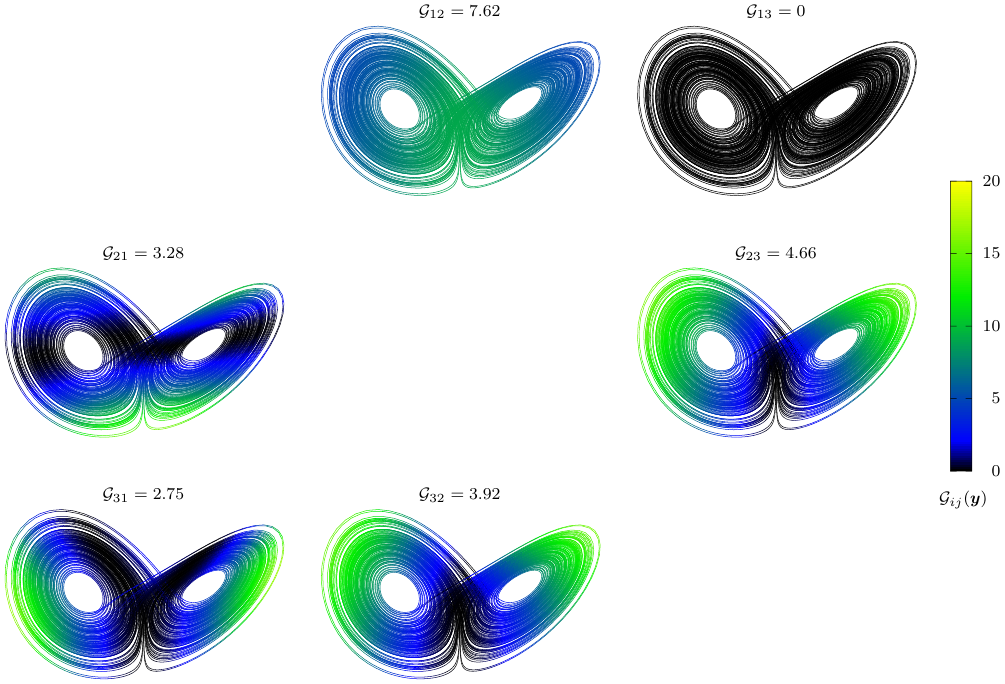}
	\end{center}
	\caption{\small GC maps of Lorenz dynamics on the attractor. The colour scale measures local Granger-causal graph values $\cG_{ij}(\by)$. The $\cG_{ij}$ values above the plots are the global Granger-causal graph values calculated according to \eqref{eq:gcint}.} \label{fig:lorenz2}
\end{figure*}

\section{Discussion}

Calculating information transfer between components of a nonlinear stochastic system, whether in the parametric sense of Wiener-Granger causality \citep{Wiener:1956,Granger:1963}, or the nonparametric sense of transfer entropy \citep{Schreiber:2000, Palus:2001}, has long been perceived as troublesome empirically and challenging, if not intractable, at the analytical level. The notion of Granger causality maps introduced by \citet{WahlEtal:2016,WahlEtal:2017} promised an elegant and powerful approach to calculation of information transfer for the important class of nonlinear stochastic dynamics that may be described by Langevin equations. The original formulation of the technique was, however, compromised by computational issues and, more seriously, by the presence of ``holes'' in the system phase space where the calculation procedure breaks down. In this paper we have presented a comprehensive resolution to both issues, so that Granger causality maps may now be considered a viable and computationally tractable method for calculation of both local and system-wide information transfer in Langevin systems. We also describe how the approach extends to calculation of information transfer for \emph{deterministic} dynamical systems described by sets of ordinary differential equations, given a template for infinitesimal fluctuations in the form of a variance-covariance matrix phase-space map.

This work opens up information flow analysis not only for standard applications of Langevin equations in statistical physics, quantum physics, noisy electrical circuits, soft-matter physics, chemical kinetics, biophysics, neural systems and econometrics, but extends the ambit of the analysis to deterministic dynamics in Hamiltonian systems, Newtonian/relativistic dynamics, classical dynamical systems and chaos theory (\cf~Section~\ref{sec:lorenz}), and also to artificial neural networks in brain modelling (\eg, neural mass models), and as increasingly deployed in machine-learning and AI applications. Taking the latter case as an illustration of potential for application, we propose that analysis of information transfer within and between layers in deep-learning, predictive coding and transformer architectures may assist in understanding the mechanisms underlying the functionality achieved by such systems---to illuminate the ``black box'', as such systems are frequently regarded---and to inform network design and AI safety issues \citep{BereskaGavves:2024}.

\ \\\noindent\textit{Limitations:} A question we have not addressed in this study is the relationship between the global GC rates obtained by averaging local GC rates (see remarks at the end of Section~\ref{sec:lvougc}) and the global transfer entropy. While GC maps may be viewed as furnishing a more fine-grained picture of information transfer between system components through localisation in phase space, we do not know whether, for a given system, the averaged local information transfer agrees with the corresponding transfer entropy. As it stands, we are only assured of a positive answer to this question in the case of linear dynamics.

A limitation for the case of neural systems, both natural and technological, is that such systems in general (certainly in the biological case) feature finite signal propagation delays in information transmission between network nodes. This may well apply too for some financial systems \citep{ComteRenault:1996}. Neither standard Langevin equations, nor indeed Ornstein-Uhlenbeck equations, accommodate finite lags \citep{BarnettSeth:ctgc:2017}. Incorporating distributed lags could in principle be achieved through generalised Langevin equations \citep{Kawasaki:1973}. While this significantly complicates the analysis, it may well be tractable to some differential-delay methods introduced in \citet{BarnettSeth:ctgc:2017}, in combination with state-space methods \citep{BarnettSeth:ssgc:2015,Solo:2016}.

\ \\\noindent\textit{Future directions:} We have not addressed deployment of GC maps for inference of information transfer from empirical time-series data. While perhaps of reduced importance in physics, which tends to proceed from a theoretical standpoint (\ie, the dynamical equations are motivated on theoretical grounds), this is of greater importance for analysis of biophysical (especially neural) and econometric systems, where there are frequently weaker theoretical motivations for detailed specification of dynamics, thus entailing the identification of an appropriate model from the data. Here we merely remark that there is a well-developed literature on methods for estimating Langevin models from discretely-sampled time-series data \citep{HindriksEtal:2011,TabarRahimi:2019,LinEtal:2025}; GC maps may then be calculated by the method described here.

A significant feature of (parametric) Granger causality is that the time-domain GC may be decomposed in the spectral domain \citep{Geweke:1982}. This decomposition is especially pertinent to functional analysis of biological neural systems, where functional (statistical) relations are frequently associated with specific frequency bands. This spectral decomposition extends, in principle, to the (zero-horizon) continuous-time case, and to VOU processes in particular. While tractable in the unconditional case \citep[Sec.~3.3, eq.~62]{BarnettSeth:ctgc:2017}, the conditional case is more mathematically intricate \citep{Geweke:1984}, and requires further work to extend to continuous time.

Finally, information transfer also forms the basis of some current theories  in the recently re-invigorated field of emergence theory in complex systems \citep{BarnettSeth:DI:2023}. Here typical systems featuring emergence, \eg, stochastic and deterministic models for flocking/swarming, are potential targets for analysis of emergence using Granger causality maps.

\subsection*{Software resources\label{sec:softrec}}
MATLAB\textsuperscript{\textregistered} code implementing calculation of Granger causality rates for VOU processes may downloaded from\\ \url{https://github.com/lcbarnett/VOUGC}.


\subsection*{Acknowledgments\label{sec:ack}}

This project has received funding from the European Research Council (ERC) under the European Union's Horizon 2020 research and innovation programme (grant agreement no. 101019254, project CONSCIOUS).


\appendix

\section*{Appendices}

\numberwithin{equation}{section}
\renewcommand{\theequation}{\thesection\arabic{equation}}

\section{Stationarity of the finite-horizon prediction error\label{apx:fpestat}}

In discrete time, consider a VMA process
\begin{equation}
	\by_k = \sum_{\ell=0}^k B_\ell \beps_{k-\ell}\,, \qquad k = 0,1,2,\ldots\,, B_0 = I, \qquad  \beps_k \sim \mathcal{N}(0,\Sigma)\, \label{eq:dtMA}
\end{equation}
with white-noise innovations process $\beps_k$. We don't assume stability, so the $B_\ell$ are not assumed square-summable, and the process $\by_k$ may not be assumed stationary; thus in contrast to \citet{BarnettSeth:ctgc:2017}, rather than an infinite past we take $\by_k$ as initialised at time step $k = 0$, with $\by_0 = \beps_0$.

We show firstly that without restrictions on the $B_\ell$, \eqref{eq:dtMA} may always be inverted to yields a (possibly unstable) VAR representation for the process $\by_k$ initialised at $k = 0$. To see this, let us set $\bY_k = \big[\by_0^\trop,\ldots,\by_k^\trop\big]^\trop$ and $\bE_k = \big[\beps_0^\trop,\ldots,\beps_k^\trop\big]^\trop$. From \eqref{eq:dtMA} we then have
\begin{equation}
	\bY_k = \bB_k \cdot \bE_k \qquad k = 0,1,2,\ldots\,,
\end{equation}
where
\begin{equation}
	\bB_k = \begin{bmatrix}
		I       & 0      & \cdots & 0 & 0     \\
		B_1     & I      & \cdots & 0 & 0      \\
		\vdots  & \vdots & \ddots & \vdots & \vdots \\
		B_{k-1} & B_{k-2}& \cdots & I & 0 \\
		B_k     & B_{k-1}& \cdots & B_1 & I
	\end{bmatrix}\,, \qquad k = 0,1,2,\ldots
\end{equation}
Since $\bB_k$ is a lower-triangular block-Toeplitz matrix with identity matrices on the diagonal, $|\bB_k| = 1$ so the matrix $\bB_k$ is invertible, and we have
\begin{equation}
	\bE_k = \bB_k^{-1} \cdot \bY_k \qquad k = 0,1,2,\ldots\,.
\end{equation}
Setting $\bA_k$ = $I-\bB_k^{-1}$, we have
\begin{equation}
	\bY_k = \bA_k \cdot \bY_k + \bE_k \qquad k = 0,1,2,\ldots\,. \label{eq:bvar}
\end{equation}
$\bA_k$ is also lower-triangular block-Toeplitz, now with zero matrices on the diagonal; it is easy to see that \eqref{eq:bvar} thus specifies a (not necessarily stable) VAR representation for the process $\by_k$, initialised at $k = 0$.

By a standard result, the optimal least-squares predictor $\hat\by_k(m) = \exop[\by_{k+m} \vbar \by_0,\ldots,\by_k]$ for $\by_{k+m}$ given $\by_0,\ldots,\by_k$ is given by \citep{Hamilton:1994}
\begin{equation}
	\hat\by_k(m) = \sum_{\ell=m}^{k+m} B_\ell \beps_{k+m-\ell} \,, \quad k = 0,1,2,\ldots\,, m = 1,2,\ldots
\end{equation}
The prediction error $\be_k(m) = \hat\by_k(m) - \by_{k+m}$ is then
\begin{equation}
	\hat\be_k(m) = -\sum_{\ell=0}^{m-1} B_\ell \beps_{k+m-\ell} \,,
\end{equation}
which, we may check, has covariance matrix
\begin{equation}
	\cE_m = \sum_{\ell=0}^{m-1} B_\ell \Sigma B_\ell^\trop \,.
\end{equation}
The crucial point to note is that this expression does not depend on the process time step $k$, so that the prediction error process is wide-sense stationary.

The above goes through in the continuous-time limit. Suppose given a continuous-time VMA process
\begin{equation}
	\by(t) = \int_{u=0}^t B(u) \,d\bw(t-u)\,, \qquad t \ge 0\,, \quad B(0) = I,  d\bw(t) \sim \mathcal{N}(0,\Sigma\,dt) \label{eq:ctMA1}
\end{equation}
where the VMA kernel $B(u)$ is no longer assumed square-integrable, initialised at $t = 0$. Suppose now given a stroboscopic sampling $\oby_k = \by(k\dt)$ with period $\dt$. Working to lowest order in $\dt$, we have
\begin{align*}
	\oby_k
	&= \int_{u=0}^{(k+1)\dt} B(u) \,d\bw(k\dt-u) \\
	&= \sum_{\ell=0}^k \int_{u=\ell \dt}^{(\ell+1)\dt} B(u) \,d\bw(k\dt-u) \\
	&\approx \sum_{\ell=0}^k B(\ell \dt) \int_{u=\ell \dt}^{(\ell+1)\dt} \,d\bw(k\dt-u) \\
	&\approx \sum_{\ell=0}^k B(\ell \dt) \,\obeps_{k-\ell}
\end{align*}
where, by the independent increments property of the Wiener process, the $\obeps_k$ are iid $\sim \mathcal{N}(0,\Sigma\,\dt)$. Thus $\oby_k$ has a VMA representation $\oby_k = \sum_{\ell=0}^k \oB_\ell \obeps_{k-\ell}$ with $\oB_k = B(k\dt)$ to lowest order in $\dt$. Taking the limit $\dt \to 0$ (\cf~\citet[Sec.~3.2]{BarnettSeth:ctgc:2017}), we see that, as in the discrete-time case, \eqref{eq:ctMA1} may always be inverted to yields a (possibly unstable) continuous-time VAR representation for the process $\by(t)$ initialised at $t = 0$, and for $h > 0$ the optimal finite-horizon least-squares prediction $\hat\by(t;h)$ of $\by(t+h)$ given $\{\by(u) : 0 \le u \le t\}$ is
\begin{equation}
	\hat\by(t;h) = \int_{u=h}^{t+h} B(u) \,d\bw(t+h-u)\,.
\end{equation}
The prediction error $\be(t;h) = \hat\by(t;h) - \by(t+h)$ is given by
\begin{equation}
	\be(t;h) = -\int_{u=0}^h B(u) \,d\bw(t+h-u)\,.
\end{equation}
By an It\=o isometry \citep{Oksendal:2003}, we find that eq.~\eqref{eq:cE} still obtains. Thus, as for the discrete-time case, $\cE(h) = \exop\big[\be(t;h)\be(t;h)^\trop\big]$ does not depend on the time stamp $t$, so that the prediction error process $\be(t;h)$ is again wide-sense stationary.

\section{Proof of Theorem~\ref{the:main}\label{apx:vougcproof}}

As in Appendix~\ref{apx:fpestat}, we sample the VOU process \eqref{eq:vou}) stroboscopically with period $\dt$. The process is not assumed stable, so that $A$ may have eigenvalues in the right half-plane of the complex plane. Using an overbar to denote quantities associated with the subsampled process (with an implied dependency on $\dt$), we find \citep[Section 3.2]{BarnettSeth:ctgc:2017} that the process $\oby_k = \by(k\dt)$ is VAR($1$)\footnote{Although \citet{BarnettSeth:ctgc:2017} consider only the stable case, the result may be shown to hold in the unstable case.}:
\begin{equation}
	\oby_k =\oA \oby_{k-1} + \obeps_k \label{eq:voux}
\end{equation}
where $\obeps_k$ is a white noise process with covariance matrix $\oSig$, and
\begin{subequations}
\begin{align}
	\oA   &= e^{A\dt} \label{eq:subvar1a} \\
	\oSig &= \Omega - e^{A\dt} \Omega e^{A^\trop \dt} \label{eq:subvar1s}
\end{align} \label{eq:subvar1}%
\end{subequations}
with $\Omega$ the solution of the continuous-time Lyapunov equation
\begin{equation}
	A \Omega + \Omega A^\trop = -\Sigma \label{eq:ctlyap}
\end{equation}
From \eqref{eq:subvar1} and \eqref{eq:ctlyap} we may calculate that to $\bigO\dt$
\begin{subequations}
\begin{align}
	\oA &= I+A\dt \label{eq:subvar1xa} \\
	\oSig &= \Sigma\label{eq:subvar1xb} \dt
\end{align} \label{eq:subvar1x}%
\end{subequations}
To calculate the reduced quantities $\Sigma^\rrop, B^\rrop(u)$ and thence $D^\rrop$, we shall work through the state-space method detailed in \citet{BarnettSeth:ssgc:2015} for the VAR($1$) process with parameters \eqref{eq:subvar1x}, and then take the limit $\dt \to 0$.

The VAR($1$) \eqref{eq:voux} is equivalent to the discrete-time innovations-form state-space (ISS) model
\begin{subequations}
\begin{align}
	\obz_{k+1} &= \oA \obz_k + \obeps_k \\
	\oby_k     &= \oA \obz_k + \obeps_k
\end{align} \label{eq:issdisc}%
\end{subequations}
with state variable $\obz_k$. The reduced state-space system---now no longer in innovations form---is then given by
\begin{subequations}
\begin{align}
	\obz_{k+1} &= \oA \obz_k + \obeps_k \\
	\oby_{\rrop,k} &= \oA_{\rrop*} \obz_k + \obeps_{\rrop,k}
\end{align} \label{eq:rssdisc}%
\end{subequations}
where subscript $*$ denotes the full indices $(1,2,3)$, and $\rrop$ the reduced indices $(1,2)$. Following \citet{BarnettSeth:ssgc:2015}, an ISS model for $\oby_{\rrop,k}$ may be derived by solution of the discrete-time algebraic Riccati equation (DARE)
\begin{equation}
	\oP = \oA\oP\oA^\trop + \oSig -  \big(\oA\oP\oA_{\rrop*}^\trop+\oSig_{*\rrop}\big) \big(\oA_{\rrop*}\oP\oA_{\rrop*}^\trop+\oSig_{\rrop\rrop}\big)^{-1} \big(\oA\oP\oA_{\rrop*}^\trop+\oSig_{*\rrop}\big)^\trop \label{eq:dare}
\end{equation}
Specifically, if a unique stabilising solution $\oP$ of \eqref{eq:dare} exists, then $\oby_{\rrop,k}$ satisfies the innovations-form SS model
\begin{subequations}
\begin{align}
	\obz_{k+1} &= \oA \obz_k + \oK^\rrop \obeps^\rrop_k \label{eq:irssdisca} \\
	\oby_{\rrop,k}  &= \oA_{\rrop*} \obz_k + \obeps^\rrop_k \label{eq:irssdiscb}
\end{align} \label{eq:irssdisc}%
\end{subequations}
where $\obeps^\rrop_k$ is a Gaussian white noise innovations process (note that this will in general \emph{not} be the same process as $\obeps_{\rrop,k}$). The covariance matrix $\oSig^\rrop$ of $\obeps^\rrop_k$ and the Kalman gain matrix $\oK^\rrop$ are given respectively by
\begin{subequations}
\begin{align}
	\oSig^\rrop &= \oA_{\rrop*}\oP\oA_{\rrop*}^\trop+\oSig_{\rrop\rrop} \\
	\oK^\rrop &= \big(\oA\oP\oA_{\rrop*}^\trop+\oSig_{*\rrop}\big) \big[\oSig^\rrop\big]^{-1}
\end{align}
\end{subequations}
\citet[Section~B]{GutknechtBarnett:2023} show, furthermore, that $\oP_{33}$ is the only non-zero block in $\oP$, and is the unique stabilising solution (if it exists) of the lower-dimensional DARE
\begin{equation}
	\oP_{33} = \oA_{33} \oP_{33} \oA_{33}^\trop + \oSig_{33} - \big(\oA_{33} \oP_{33} \oA_{\rrop3}^\trop + \oSig_{3\rrop}\big) \big(\oA_{\rrop3} \oP_{33} \oA_{\rrop3}^\trop + \oSig_{\rrop\rrop}\big)^{-1} \big(\oA_{33} \oP_{33} \oA_{\rrop3}^\trop + \oSig_{3\rrop}\big)^\trop \label{eq:sdare}
\end{equation}
The discrete-time reduced innovations covariance matrix and Kalman gain matrix then become respectively
\begin{subequations}
\begin{align}
	\oSig^\rrop &= \oA_{\rrop3} \oP_{33} \oA_{\rrop3}^\trop + \oSig_{\rrop\rrop} \\
	\oK^\rrop &=
	\begin{bmatrix}
		I \\
		\oKay^\rrop
	\end{bmatrix}
\end{align} \label{eq:dtrpssp}%
\end{subequations}
where
\begin{equation}
	\oKay^\rrop = \big(\oA_{33} \oP_{33} \oA_{\rrop3}^\trop + \oSig_{3\rrop}\big)\big[\oSig^\rrop\big]^{-1}  \label{eq:dtrpsspK}
\end{equation}

We now consider the limit $\dt \to 0$. Letting $P_{33} = \lim_{\dt \to 0} \oP_{33}$, we find that $P_{33}$ satisfies the continuous-time algebraic Riccati equation (CARE)
\begin{equation}
	A_{33}P_{33} + P_{33}A_{33}^\trop + \Sigma_{33} = \big(P_{33} A_{\rrop3}^\trop + \Sigma_{3\rrop}\big) \Sigma_{\rrop\rrop}^{-1} \big(P_{33} A_{\rrop3}^\trop + \Sigma_{3\rrop}\big)^\trop \label{eq:scarex}
\end{equation}
(we assume from here on that $\Sigma_{\rrop\rrop}$ is positive-definite, and hence invertible). Given positive-semidefinite $\Sigma$, sufficient conditions for existence of a unique stabilising solution of \eqref{eq:scarex} are \citep{ArnoldLaub:1984,NiMaoLin:2008}:
\begin{enumerate} [label={S\arabic*}]
	\item \hspace{-5pt}:\ $(A_{33},\, A_{\rrop3})$ detectable, and \label{it:carea1}
	\item \hspace{-5pt}:\ $(A_{33}-\Sigma_{3\rrop}\Sigma_{\rrop\rrop}^{-1}A_{\rrop3},\, \Sigma_{33}-\Sigma_{3\rrop}\Sigma_{\rrop\rrop}^{-1}\Sigma_{\rrop3})$ stabilisable. \label{it:carea2}
\end{enumerate}
If $\Sigma$ is positive-definite, then \ref{it:carea2} is trivially satisfied. If $A_{33}$ is stable, then \ref{it:carea1} is trivially satisfied, and in Appendix~\ref{apx:PBH} we show that \ref{it:carea1} is satisfied if $A$ is stable.  Assuming \ref{it:carea1} and \ref{it:carea2}, we have
\begin{subequations}
\begin{align}
	\Sigma^\rrop &= \lim_{\dt \to 0} \oSig_{\rrop\rrop} = \Sigma_{\rrop\rrop} \\
	K^\rrop &= \lim_{\dt \to 0} \oK^\rrop =
	\begin{bmatrix}
		I \\
		\Kay^\rrop
	\end{bmatrix} \label{eq:kayR}
\end{align}%
\end{subequations}
with
\begin{equation}
	\Kay^\rrop = \big(P_{33} A_{\rrop3}^\trop + \Sigma_{3\rrop}\big) \Sigma_{\rrop\rrop}^{-1} \label{eq:kalred}
\end{equation}

Now let $\oby_{\rrop,k} = \by_\rrop(k\dt)$ be the discretised reduced process \eqref{eq:irssdiscb}. As demonstrated in Appendix~\ref{apx:fpestat}, to lowest order in $\dt$ the VMA coefficient matrices of $\oby_{\rrop,k}$ are $\oB^\rrop_k = B^\rrop(k\dt)$, where $B^\rrop(u)$ is the VMA kernel of the continuous-time reduced process $\by_\rrop(t)$. From \eqref{eq:rssdisc}, it is easily calculated that $\oB^\rrop_0 = I$ and
\begin{equation}
	\oB^\rrop_k = \oA_{\rrop*}\oA^{k-1}\oK^\rrop\,, \qquad k > 0\,.\label{eq:ssBk}
\end{equation}
Let us fix $u > 0$. Making explicit the dependency of $\oB^\rrop_k, \oA^\rrop, \oK^\rrop$, \etc, on $\dt = u/k$, we have
\begin{align*}
	& B^\rrop(u) \\
	&= \lim_{k \to \infty}\; \oB^\rrop_k(u/k) \\
	&= \lim_{k \to \infty}\; \oA_{\rrop*}(u/k) \cdot \oA(u/k)^{k-1} \cdot \oK^\rrop(u/k) \quad\text{from \eqref{eq:ssBk}} \\
	&= \lim_{k \to \infty}\; [I+ A_{\rrop\rrop}(u/k) \ A_{\rrop3}(u/k)] \cdot [I+A(u/k)] \cdot [I+A(u/k)]^k \cdot \begin{bmatrix} I \\ \Kay^\rrop \end{bmatrix} \quad\text{from \eqref{eq:subvar1xa} and \eqref{eq:kayR}} \\
	&= [I \ 0] \cdot e^{Au} \cdot \begin{bmatrix} I \\ \Kay^\rrop \end{bmatrix} \,,
\end{align*}
where in the last step we use $\lim_{k \to \infty}\; [I+A(u/k)]^k = e^{Au}$. In particular, we have
\begin{equation}
	\dot B^\rrop(0) = [I \ 0]
	\begin{bmatrix} A_{\rrop\rrop} & A_{\rrop3} \\ A_{3\rrop} & A_{33} \end{bmatrix}
	\begin{bmatrix} I \\ \Kay^\rrop \end{bmatrix}
	= A_{\rrop\rrop} + A_{\rrop3} \Kay^\rrop \,. \label{eq:BR}
\end{equation}
Putting all this together, from \eqref{eq:ddef}, \eqref{eq:D1}, \eqref{eq:kalred} and \eqref{eq:BR}, some straightforward algebra yields
\begin{equation}
	D^\rrop-D_{\rrop\rrop} = A_{\rrop 3} P_{33} A_{\rrop 3}^\trop
\end{equation}
so that in particular
\begin{equation}
	D^\rrop_{11}-D_{11} = A_{13} P_{33} A_{13}^\trop
\end{equation}
and from \eqref{eq:gcrate} we have
\begin{equation}
	\cR_{\by_3 \to \by_1|\by_2} = \traceop\!\Big[\Sigma^{-1}_{11} A_{13} P_{33} A_{13}^\trop\Big]
\end{equation}
as stated in Theorem~\ref{the:main}.

Finally, for a (not necessarily stable) CTVAR process $\by(t)$ initialised at $t = 0$, we define the transfer entropy rate analogously to the GC rate as
\begin{equation}
	\cT_{\by_3 \to \by_1|\by_2} = \lim_{h \to 0} \,\frac1h \cT_{\by_3 \to \by_1|\by_2}(h)\,,
\end{equation}
with
\begin{equation}
	\cT_{\by_3 \to \by_1|\by_2}(h) = \enop[\by_1(t+h) \vbar \by_\rrop(u) : 0 \le u \le t] - \enop[\by_1(t+h) \vbar \by(u) : 0 \le u \le t] \label{eq:ctte}
\end{equation}
the TE at finite prediction horizon $h > 0$, where $\enop[\cdots\vbar\cdots]$ denotes conditional differential entropy. This is equivalent to the continuous-time TE rate as set out in \citet{CooperEdgar:2019}. Note now, that for multivariate-normal variables $\by,\bx$, the conditional differential entropy $\enop[\by\vbar\bx]$ is given by $\enop[\beps]$, where $\beps \sim \cN(0,\Sigma)$ is the residual error of the projection $\exop[\by\vbar\bx]$ of $\by$ on $\bx$---\ie, $\beps  = \by-\exop[\by\vbar\bx]$---and we have\footnote{Up to an additive constant that depends only on the dimension of the variable $\by$.} $\enop[\beps] = \frac12\log|\Sigma|$.  By a similar stroboscopic discretisation argument to that in Appendix~\ref{apx:fpestat}, and since increments of the Wiener noise $\bw(t)$ are multivariate-normal, the conditional entropies on the right-hand side of \eqref{eq:ctte} are seen to be $\frac12\log|\cE^\rrop_{11}(h)|$ and $\frac12\log|\cE_{11}(h)|$ respectively\footnote{This confirms that, as our notation suggests,  the expression for $\cT_{\by_3 \to \by_1|\by_2}(h)$ in \eqref{eq:ctte} does not depend on the time stamp $t$.}, with $\cE(h)$ as in \eqref{eq:cE} and $\cE^\rrop(h)$ its reduced-process counterpart. Thus from \eqref{eq:GCh} we have $\cF_{\by_3 \to \by_1 | \by_2}(h) = 2\cT_{\by_3 \to \by_1|\by_2}(h)$, and from \eqref{eq:GCrdef} $\cF_{\by_3 \to \by_1 | \by_2} = 2\cT_{\by_3 \to \by_1|\by_2}$ as stated in Theorem~\ref{the:main}.

\section{Proof of detectability for stable full matrix\label{apx:PBH}}

Given an $n \times n$ matrix
\begin{equation}
	A = \begin{bmatrix} A_{11} & A_{12} \\ A_{21} & A_{22} \end{bmatrix},
\end{equation}
we show that if $A$ is Hurwitz-stable then $(A_{11},A_{21})$ is detectable.

We use a Popov-Belevitch-Hautus (PBH) test \citep{Kailath:1980}. Suppose that $(A_{11},A_{21})$ is \emph{not} detectable. This implies that $A_{11}$ has an eigenvalue $\lambda$ with $\re(\lambda) \ge 0$ such that
\begin{equation}
	\rank\begin{bmatrix} \lambda I-A_{11} \\ A_{21} \end{bmatrix} < n_1\,,
\end{equation}
where $n_1$ is the size of $A_{11}$, so that there exists a nonzero vector $\bv$ such that
\begin{equation}
	\begin{bmatrix} \lambda I-A_{11} \\ A_{21} \end{bmatrix} \bv = 0\,,
\end{equation}
or
\begin{subequations}
\begin{align}
	A_{11}\bv &= \lambda\bv\,, \\
	A_{21}\bv &= 0\,.
\end{align} \label{eq:Av}
\end{subequations}
Now
\begin{align*}
	A\begin{bmatrix}\bv\\0\end{bmatrix}
	&= \begin{bmatrix}A_{11}\bv\\A_{21}\bv\end{bmatrix} \\
	&= \begin{bmatrix}\lambda\bv\\0\end{bmatrix} \qquad\text{by \eqref{eq:Av}} \\
	&= \lambda\begin{bmatrix}\bv\\0\end{bmatrix}
\end{align*}
so that $\lambda$ is an eigenvalue of $A$. But by assumption $\lambda$ is unstable, which contradicts Hurwitz-stability of $A$. Thus $(A_{11},A_{21})$ must be detectable.

%\vspace{2em}
%\bibliographystyle{plain}
%\bibliographystyle{model2-names}
%\bibliographystyle{plainnat}
%\bibliographystyle{apsrev4-2}
\bibliographystyle{apalike}

\bibliography{vougc}

\end{document}